\documentclass[10pt]{article}

\usepackage[utf8]{inputenc}
\usepackage[T1]{fontenc}
\usepackage{float}
\usepackage{array}
\usepackage{color}

\usepackage{multirow}
\usepackage{epsfig}
\usepackage{amssymb}
\usepackage{amsfonts}
\usepackage{graphicx}
\usepackage{fancyhdr}
\usepackage{multirow}
\usepackage{marginnote}
\usepackage{listings}
\usepackage{float}
\usepackage[font=footnotesize]{subcaption}
\usepackage{algorithmicx}
\usepackage{algpseudocode}
\usepackage{algorithm}
\usepackage{xspace}
\errorcontextlines\maxdimen

\makeatletter
\newcommand*{\algrule}[1][\algorithmicindent]{\makebox[#1][l]{\hspace*{.5em}
\vrule height .75\baselineskip depth .25\baselineskip}}%

\newcount\ALG@printindent@tempcnta
\def\ALG@printindent{%
    \ifnum \theALG@nested>0
        \ifx\ALG@text\ALG@x@notext
            \addvspace{-3pt}
        \else
            \unskip
            \ALG@printindent@tempcnta=1
            \loop
                \algrule[\csname 
ALG@ind@\the\ALG@printindent@tempcnta\endcsname]%
                \advance \ALG@printindent@tempcnta 1
            \ifnum \ALG@printindent@tempcnta<\numexpr\theALG@nested+1\relax%
            \repeat
        \fi
    \fi
    }%
\usepackage{etoolbox}
\patchcmd{\ALG@doentity}{\noindent\hskip\ALG@tlm}{\ALG@printindent}{}{
\errmessage{failed to patch}}
\makeatother

\usepackage[normalem]{ulem}
\usepackage{hyperref}

 \usepackage{dsfont}
 \usepackage{tikz}
 
 \usetikzlibrary{snakes,positioning,shapes.geometric,arrows,arrows.meta}
 \tikzstyle{node}=[draw]
 \tikzstyle{tikz}=[draw]

\usepackage{varwidth}

\usepackage{multirow}
 \usepackage{amsmath}
 \usepackage{bbm}
 \usepackage{eufrak}

\newcommand{\picc}[1]{\Pi_{cc}^{||}(#1)}

\newcommand{\rsg}{\mathsf{RSG}_{\lambda}^k}

\usepackage{xspace}

\newcommand{\etc}{\textit{etc.}\xspace}
\newcommand{\ie}{\textit{i.e.,}\xspace}
\newcommand{\eg}{\textit{e.g.,}\xspace}

\newcommand{\GG}{{G}}
\newcommand{\E}[0]{\mathbb{E}}

\newcommand{\x}[0]{\mathbf{x}}

\newcommand{\X}[0]{\mathbf{X}}

\newcommand{\tuple}[1]{\left \langle #1 \right \rangle}

\definecolor{ltgray}{RGB}{200,200,200}
\definecolor{dkgray}{RGB}{150,150,150}

\usepackage{amsmath}

\newcommand{\rifflePass}{$\textsf{RiffleScrambler}$\xspace}
\newcommand{\traceTraj}{$\textsf{TraceTrajectories}$\xspace}

\usepackage{amsthm}
\newtheorem{lemma}{Lemma}
\newtheorem{theorem}{Theorem}
\newtheorem{definition}{Definition}
\newtheorem{example}{Example}
\newtheorem{corollary}{Corollary}
\newtheorem{remark}{Remark}

\begin{document}

\title{
RiffleScrambler -- a memory-hard password storing function 
\thanks{
 Authors were supported by Polish 
   National Science Centre contract number DEC-2013/10/E/ST1/00359.
 }\thanks{Accepted for ESORICS 2018}
}
\author{ 
 Karol Gotfryd\\{\small Wrocław University of}\\{\small Science and 
Technology}\and
 Paweł Lorek\\{\small Wrocław University} \and 
 Filip Zagórski\footnote{Corresponding author}\\{\small Wrocław University 
of}\\{\small Science and Technology,}\\{\small and Oktawave}
 }

\date{}
 
\maketitle

\begin{abstract}
We introduce RiffleScrambler: a new family of directed acyclic graphs and a 
corresponding  data-independent memory hard function with password independent 
memory access. We prove its memory hardness in the random oracle model.

RiffleScrambler is similar to Catena -- updates of hashes are determined by a 
graph (bit-reversal or double-butterfly graph in Catena). The advantage of the 
RiffleScrambler over Catena is that the underlying graphs are not predefined but 
are generated per salt, as in Balloon Hashing. Such an approach leads to higher 
immunity against practical parallel attacks. RiffleScrambler offers better 
efficiency than Balloon Hashing since the in-degree of the underlying graph is 
equal to 3 (and is much smaller than in Ballon Hashing). 
At the same time, because the underlying graph is an instance of a 
Superconcentrator, our construction achieves the same time-memory trade-offs.

\vspace{.25cm}
\end{abstract}

\section{Introduction}\label{sec:introduction}

In early days of computers' era passwords were stored in plaintext in 
the form of pairs $(user, password)$. Back in 1960s it was observed,
that it is not secure. It took around a decade to incorporate a more secure 
way of storing users' passwords -- via a DES-based function 
\texttt{crypt}, as $(user, 
f_k(password))$ for a secret key $k$ or as $(user, f(password))$ for a one-way 
function. The first approach (with encrypting passwords) allowed admins to 
learn a user password while both methods enabled to find two users with the 
same password. To mitigate that problem, a random \textsf{salt} was added and 
in most systems, 
passwords are stored as  $(user, f(password, salt), salt)$ for a one-way 
function $f$. 
This way, two identical passwords will with high probability
be stored as
two  different strings (since the salt
should be chosen uniformly at random).

Then another issue appeared: an adversary, having a database of hashed 
passwords, can try many dictionary-based 
passwords or a more targeted attack~\cite{Wang2016}. The ideal password should 
be random, but users tend to select 
passwords with low entropy instead. Recently, a probabilistic model of the 
distribution of choosing passwords, based on Zipf's law  was 
proposed~\cite{Wang2017}. That is why one requires  
 password-storing function to  be ``\textbf{slow}'' enough to compute 
for an attacker and ``\textbf{fast}'' enough to compute for the 
authenticating server.

To slow-down external attackers two additional enhancements were proposed.
\textsf{Pepper} is similar to salt, its value is sampled from uniform 
distribution and passwords are stored as $(user, f(password, salt, pepper), 
salt)$ but the value of \textsf{pepper} is not stored -- evaluation (and 
guessing) of a password is slowed down by the factor equal to the size of the 
space from which \textsf{pepper} is chosen. \textsf{Garlic} is also used to 
slow down the process of verifying a password -- it tells how many 
times $f$ is called, a password is stored as $(user, 
f^{garlic}(password, salt, pepper), salt, garlic)$. Using \textsf{pepper} and 
\textsf{garlic} has one bottleneck -- it slows down at the same rate both 
an attack and the legitimate server.
Moreover this approach does not guarantee required advantage over 
an adversary who can evaluate a function simultaneously using ASICs 
(Application 
Specific Integrated Circuit) or clusters of GPUs. The reason is that 
the cost of evaluation of hash functions like SHA-1, SHA-2 on an ASIC is even 
several thousands times smaller than on a CPU. For instance Antminer S9 has 
claimed performance of 14 TH/s~\cite{antminer} (double SHA256) versus about 30 
GH/s for 
the best currently available GPU cards (GeForce GTX 1080TI, Radeon RX Vega 56) 
and up to 1 GH/s for CPUs.
Percival~\cite{Percival} noted that memory cost is more comparable across 
various platforms than computation time. He suggested a memory-hard functions 
(MHF) and introduced~\texttt{scrypt}~\cite{Percival}.
So the right approach in designing a password-storing function is not only to 
make a function slow, but to make it use 
relatively large amounts of memory.

One can think about an MHF as a procedure of evaluating a function $F$ which 
uses some ordering of memory calls. 
Such a function can be described as a DAG (Directed Acyclic Graph) $G = G_F$. 
Vertex $v$ represents some internal value,
and if evaluating this value requires values at $v_{i_1},\ldots,v_{i_M}$, then 
$v_{i_1},\ldots,v_{i_M}$ are parents of $v$ in $G$.
For example, calculating a function $ F = H^N(x) = H(\ldots H(H(x))\ldots )$ 
($H$ 
computed $N$ times -- like PBKDF2 -- Password Based Key Derivation 
Function~2~\cite{pbkdf} where $N = 1024$) can be represented by a graph 
$G=(V,E)$ with vertices $\{v_0,v_1,\ldots, v_{N-1}, v_N\}$ and edges 
$E=\{(v_i,v_{i+1}), i=0,\ldots,N-1\}$ (\ie a path). Initially, the value 
at vertex $v_0$ is $x$, computations yield in value $F(x)$ at vertex $v_{N}$.

We can categorize MHFs into two groups: data dependent MHFs (dMHF) 
and data independent MHFs (iMHF). 
Roughly speaking, for iMHFs the order of computations (graph $G_F$)  does not 
depend on a \textsl{password} (but it can still depend on a 
\textsl{salt}), whereas the ordering of calculations in dMHFs depends on data 
(\ie a \textsl{password}).
Because dMHFs may be susceptible to various side-channel attacks 
(via \eg timing attacks, memory access pattern) the main focus is on 
designing iMHFs.

\paragraph{Sequential vs Parallel attacks.}
Security of memory-hard functions can be analyzed in two models: sequential one 
and parallel one. In the sequential model, an adversary tries to invert a 
function by performing computation on a single-core machine while in the 
parallel model an adversary may use many processors to achieve its goal.
The security of a given memory-hard function comes from the properties of the 
underlying graphs, \eg if an underlying graph is a {\em{Superconcentrator}} 
(see Definition~\ref{def:superc}) 
then a function is memory-hard in the sequential model
while if the underlying graph is {\em{depth-robust}} then the function is 
memory-hard in the parallel model.

\subsection{Related Work}
To compare various constructions, besides iMHF and dMHF distinction,
one considers the complexity of evaluation of a given function. 
The formal definitions of sequential/parallel complexity are stated in
Section~\ref{sec:prelim} (and are formalized as a pebble-game), here we provide 
the intuitive motivation behind
these definitions. The sequential complexity $\Pi_{st}^{}(G)$ of a 
directed acyclic graph $G$
is the time it takes to label (pebble/evaluate) the graph times the maximal 
number of memory cells the best sequential 
algorithm   needs to 
evaluate (pebble) the graph.
In the similar fashion one defines cumulative complexity of 
parallel pebbling $\Pi_{cc}^{||}(G)$ of $G$, here this is the sum of the number
of used memory cells during  labeling (pebbling) the graph by the best
parallel algorithm. For detailed discussion on the pebbling game see e.g. 
\cite{Alwen2015,Boneh2016,Forler,Lengauer1982}.

For a graph that corresponds to evaluating \texttt{PBKDF2} above values
are equal to $\Pi_{st}(PBKDF2) = n$ and $\Pi_{cc}^{||}(PBKDF2) = n$
while if a function is memory hard $\Pi_{st}(PBKDF2) = \Omega(n^2)$.

The Password Hashing 
Competition~\cite{phc} was announced in 2013 in the quest of 
a standard function in the area. \texttt{Argon2i}~\cite{Biryukov2016} was the 
winner, while \texttt{Catena}~\cite{Forler2015_Catena}, 
\texttt{Lyra2}~\cite{Simplicio}, \texttt{yescript}~\cite{yescrypt} and 
\texttt{Makwa}~\cite{Pornin2015} were given a special recognition.

Let $\textsf{DFG}_n^\lambda$ be the Catena Dragonfly Graph,  
$\textsf{BFG}_n^\lambda$ be the Catena Butterfly Graph (see 
\cite{Alwen2017,Forler2015_Catena})
and $\textsf{BHG}_\sigma^\lambda$ be the graph corresponding
to evaluating Balloon Hashing~\cite{Boneh2016}.

\noindent
Concerning sequential attacks for $\textsf{DFG}_N^\lambda$ and 
$\textsf{BHG}_\sigma^\lambda$ we have:
%
\begin{lemma}\label{lem:result_dfg}
Any adversary using $S \leq 
N/20$  memory cells requires $T$ placements such that
\[
	T \geq N \left(\frac{\lambda N}{64 S}\right)^{\lambda}
\]
for $\textsf{DFG}_N^\lambda$.
\end{lemma}

\begin{lemma}\label{lem:result_balloon}
Any adversary using $S \leq N/64$ memory cells, for in-degree $\delta = 7$ and 
$\lambda$ rounds requires
\[
 T \geq \frac{(2^{\lambda} - 1) N^2}{32 S}
\]
placements for $\textsf{BHG}_\sigma^\lambda$.
 
\end{lemma}

\noindent
Concerning parallel attacks for $\textsf{BFG}^\lambda_n$, $\textsf{DFG}^n_\lambda$ and $\textsf{BHG}_\sigma^\lambda$  we have the following
\begin{theorem}[Theorem 7 in \cite{Alwen2017}]
 \begin{itemize}
  \item If $\lambda,n \in \mathbb{N}^+$ such that $n=2^g(\lambda(2g-1)+1)$ for 
some 
$g\in\mathbb{N}^+$ then 
  $$\Pi_{cc}^{||}(\textsf{BFG}_n^\lambda)=\Omega\left({n^{1.5}\over 
g\sqrt{g\lambda}}\right)$$

  \item If $\lambda,n \in \mathbb{N}^+$ such that $k=n/(\lambda+1)$ is a power 
of 2 
then 
  $$\Pi_{cc}^{||}(\textsf{DFG}_n^\lambda)=\Omega\left({n^{1.5}\over 
\sqrt{\lambda}}\right)$$
  
  \item If $\tau, \sigma\in\mathbb{N}^+$ such that $n=\sigma \cdot \tau$ then with 
high probability  
  $$\Pi_{cc}^{||}(\textsf{BHG}_\tau^\sigma)=\Omega\left({n^{1.5}\over 
\sqrt{\tau}}\right)$$  
  \end{itemize}

\end{theorem}

Alwen and Blocki~\cite{Alwen2016} show that it is possible (in parallel 
setting) that an attacker may save space for any iMHF (\eg Argon2i, Balloon, 
Catena, \etc) so the $\Pi_{cc}^{||}$ is not $\Omega(n^2)$ but 
$\Omega(n^2/{\log n})$. The attack is applicable only when the instance of an 
MHF 
requires large amount of memory (\eg~>~1GB) while in practice, MHFs would be 
run so the memory consumption is of the order of just several megabytes (\eg 
16MB).
In order to mount such an attack, a special-purpose hardware must be built 
(with lots of shared memory and many cores).
While Alwen, Blocki and Pietrzak~\cite{Alwen2017} improve 
these 
attacks even further, it was still unclear if this type of 
attack is of practical concern. But recently Alwen and Blocki~\cite{Alwen2017b} 
improved these attacks even more and presented implementation which \eg
successfully ran against the latest version of Argon (Argon2i-B).

%

\subsection{Our contribution}
In this paper (details given in Section \ref{sec:rifflePass}) we introduce 
\rifflePass~-- a new family of directed acyclic graphs and a corresponding 
data-independent memory 
hard function with password independent memory access. We prove its memory 
hardness in the random oracle model.
\smallskip\par\noindent
For a password $x$, a salt $s$ and security parameters (integers) 
$g, \lambda$:
\begin{enumerate}
 \item a permutation $\rho = \rho_{g}(s)$ of $N = 2^g$ elements is 
generated using (time reversal of) Riffle Shuffle (see Algorithm 
\ref{alg:riffle}) and
 \item  a computation graph  $\mathsf{RSG}_{\lambda}^N = 
G_{\lambda, g, s} = G_{\lambda, g}(\rho)$ is 
generated.
\end{enumerate}

Evaluation of \rifflePass (a function on $2\lambda$ 
stacked $\mathsf{RSG}^N$ graphs) with $S = N = 2^g$  memory cells 
takes the number of steps proportional to $T \approx 3 \lambda N$ (the 
max-indegree of the graph is equal to $3$).


\noindent
Our result on time-memory trade-off concerning \textsl{sequential attacks} is following.
\begin{lemma}\label{lem:result}
Any adversary using $S \leq 
N/20$  memory cells requires $T$ placements such that
\[
	T \geq N \left(\frac{\lambda N}{64 S}\right)^{\lambda}
\]
for the $\mathsf{RSG}_{\lambda}^N$.
\end{lemma}
The above lemma means that (in the sequential model) any adversary who uses $S$ 
memory cells for which $S \leq N/20$ would spend at least $T$ steps evaluating the 
function and $T \geq N \left(\frac{\lambda N}{64 S}\right)^{\lambda}$ (the punishment for decreasing available memory cells is severe).
The result for sequential model gives the same level as for Catena (see Lemma \ref{lem:result_dfg}) while it is  
much better than for BallonHashing (see Lemma \ref{lem:result_balloon}). The main advantage of \rifflePass (compared to Catena) is that each
salt corresponds (with high probability) to a different computation graph, and 
there are $N! = 2^g!$ of them (while Catena uses one, \eg bit-reversal based 
graph).
Moreover it is easy to modify \rifflePass so the number of possible computation 
graphs is equal to $N!^{2 \lambda} = (2^g!)^{2 \lambda}$ (we can use different -- also salt dependent -- permutations in each stack).
\smallskip\par
On the other hand \rifflePass guarantees better immunity against parallel 
attacks than Catena's computation graph.
Our result concerning \textsl{parallel attacks} is following.
\begin{lemma}\label{lem:main_par}
For positive integers $\lambda, g$ let $n = 2^g(2 
\lambda g + 1)$  for 
some $g \in \mathbb{N}^{+}$.
Then
\[
 \picc{\rsg} = 
\Omega\left(\frac{n^{1.5}}{\sqrt{g \lambda}}\right).
\]
\end{lemma}

The \rifflePass is a password storing method that is immune to cache-timing 
attacks since memory access pattern is password-independent.

\section{Preliminaries}\label{sec:prelim}

For a directed acyclic graph (DAG) $G = (V, E)$ of $n = |V|$ nodes, we say 
that the indegree is $\delta = \max_{v \in V} indeg(v)$ if $\delta$ is the 
smallest number such that for any $v \in V$ the number of incoming edges is not 
larger than $\delta$.
Parents of a  node $v \in V$ is the set $\textsf{parents}_G(v) = \{u \in V: (u, 
v) \in E\}$. 

We say that $u \in V$ is a \textit{source} node if it has no 
parents (has indegree 0) and we say that $u \in V$ is a \textit{sink} node if 
it is not a parent for any other node (it has 0 outdegree). We denote the set 
of all sinks of $G$ by $\textsf{sinks}(G) = \{v \in V: outdegree(v) = 0\}$.



The directed path $p = (v_1, \ldots, 
v_t)$ is of length $t$ in $G$ if $(\forall i) v_i \in V$, $(v_i, 
v_{i+1}) \in E$, we denote it by $\textsf{length}(p) = t$. The \textit{depth} 
$d = \textsf{depth}(G)$ of a graph $G$ is the length of the longest directed 
path in $G$.

\subsection{Pebbling and complexity}
One of the methods for analyzing iMHFs is to use so called \textsl{pebbling 
game}. We will follow \cite{Alwen2017} with the notation
(see also therein for more references on pebbling games).

\begin{definition}[Parallel/Sequential Graph Pebbling]\label{def:pebbling}
Let $G = (V, E)$ be a DAG and let $T \subset V$ be a target set of nodes to be 
pebbled. A pebbling configuration (of $G$) is a subset $P_i \subset V$. A legal 
parallel pebbling of $T$ is a sequence $P = (P_0, \ldots, P_t)$ of pebbling 
configurations of $G$, where $P_0 = \emptyset$ and which satisfies conditions $1$ 
and $2$ below. A sequential pebbling additionally must satisfy condition $3$.
\begin{enumerate}
 \item At some step every target node is pebbled (though not necessarily 
simultaneously).
  \[
   \forall x \in T \quad \exists z \leq t \quad : \quad x \in P_z.
  \]

  \item Pebbles are added only when their predecessors already have a pebble at 
the end of the previous step.
  \[
   \forall i \in [t] \quad : \quad x \in (P_i \setminus P_{i-1}) \Rightarrow 
\textsf{parents}(x) \subset P_{i-1}.
  \]

  \item At most one pebble is placed per step.
  \[
   \forall i \in [t] : |P_i \setminus P_{i-1}| \leq 1.
  \]

\end{enumerate}
We denote with $\mathcal{P}_{G, T}$ and $\mathcal{P}_{G, T}^{||}$ the set of 
all legal sequential and parallel pebblings of $G$ with target set $T$, 
respectively. 
\end{definition}
Note that $\mathcal{P}_{G, T} \subset \mathcal{P}_{G, T}^{||}$. The most 
interesting case is when $T = \textsf{sinks}(G)$, with such a $T$ we write 
$\mathcal{P}_G$ and $\mathcal{P}_G^{||}$.
 
\begin{definition}[Time/Space/Cumulative Pebbling Complexity]
The time, space, space-time and cumulative complexity of a pebbling $P = \{P_0, 
\ldots, P_t\} \in \mathcal{P}_G^{||}$ are defined as:
\[
 \Pi_t(P) = t, \quad \Pi_s(P) = \max_{i \in [t]}|P_i|, \quad \Pi_{st}(P) = 
\Pi_t(P) \cdot \Pi_s(P), \quad \Pi_{cc}(P) = \sum_{i\in[t]} |P_i|.
\]
For $\alpha \in \{s, t, st, cc\}$ and a target set $T \subset V$, the 
sequential and parallel pebbling complexities of $G$ are defined as
\[
 \Pi_{\alpha}(G, T) = \min_{P \in \mathcal{P}_{G, T}} \Pi_{\alpha}(P), \qquad 
 \Pi_{\alpha}^{||}(G, T) = \min_{P \in \mathcal{P}_{G, T}^{||}} 
\Pi_{\alpha}(P). 
\]
When $T = \textsf{sinks}(G)$ we write $\Pi_{\alpha}(G)$ and 
$\Pi_{\alpha}^{||}(G)$. 
\end{definition}

%
%
%

\subsection{Tools for sequential attacks}

\begin{definition}[$N$-Superconcentrator]\label{def:superc}
A directed acyclic graph $\GG = \tuple{V, E}$ with a set of vertices $V$ and a 
set of edges $E$, 
a bounded indegree, $N$ inputs, and $N$ outputs is called $~{\textsf{N-Superconcentrator}}$ if
for every $k$ such that $1 \leq k \leq N$ and for every pair of
subsets $V_1 \subset V$ of $k$ inputs and $V_2 \subset V$ of $k$ outputs, there are $k$ vertex-disjoint paths connecting 
the vertices in $V_1$ to the vertices in $V_2$.
\end{definition}
\noindent
By stacking $\lambda$ (an integer) $N$-Superconcentrators we obtain a graph called $(N,\lambda)$-Superconcentrator.
\begin{definition}[$(N, \lambda)$-Superconcentrator]
Let $\GG_i, i=0,\ldots, \lambda-1$ be $N$-Superconcentrators. Let $\GG$ be the graph created by joining the outputs 
of $\GG_i$ to the corresponding inputs of $\GG_{i+1}, i=0,\ldots,\lambda-2$. Graph $\GG$ is called 
$(N, \lambda)$-Superconcentrator.
\end{definition}

\begin{theorem}[Lower bound for a $(N, \lambda)$-Superconcentrator 
\cite{Lengauer1982}]\label{thm:superconcentrator}
Pebbling a $(N, \lambda)$-Superconcentrator using $S \leq 
N/20$   pebbles requires $T$ placements such that
\[
	T \geq N \left(\frac{\lambda N}{64 S}\right)^{\lambda}.
\]
\end{theorem}

%
%
%
%
%
%

\subsection{Tools for parallel attacks}
We build upon the results from \cite{Alwen2017}, hence we shall recall the
definitions used therein.
%

\begin{definition}[Depth-Robustness \cite{Alwen2017}]\label{def:depthRobustness}
For $n \in \mathcal{N}$ and $e, d \in [n]$ a DAG $G = (V, E)$ is $(e, 
d)$-depth-robust if 
\[
 \forall S \subset V \quad |S| \leq e \Rightarrow \textsf{depth}(G-S) \geq d.
\]
\end{definition}

\noindent
For $(e,d)$-depth-robust graph we have
\begin{theorem}[Thm. 4 ~\cite{Alwen2017}]\label{thm:thm4}
Let $G$ be an $(e, d)$-depth-robust DAG. Then we have  $\Pi_{cc}^{||}(G) > ed$.

\end{theorem}
We can obtain better bounds on $\Pi_{cc}^{||}(G)$ having more assumptions on the structure of $G$.

\begin{definition}[Dependencies \cite{Alwen2017}]\label{def:dep} Let $G=(V,E)$ be a DAG and $L\subseteq V$. We say that 
 $L$ has a $(z,g)-$dependency if there exist nodes-disjoint paths $p_1,\ldots,p_z$ of length at least $g$ each ending in $L$.
\end{definition}

\begin{definition}[Dispersed Graph \cite{Alwen2017}] Let $g\geq k$ be positive integers. A DAG $G$ is called $(g,k)-$dispersed
if there exists ordering of its nodes such that the following holds. Let $[k]$ denote last $k$ nodes in the 
ordering of $G$ and let $L_j=[jg,(j+1)g-1]$ be the $j^{th}$ subinterval. Then $\forall j\in [\lfloor k/g \rfloor]$ the interval 
$L_j$ has a $(g,g)-$dependency. More generally, let $\varepsilon\in(0,1]$. If each interval $L_j$ 
only has an $(\varepsilon g,g)-$dependency, then $G$ is called $(\varepsilon, g,k)$-dispersed.
\end{definition}

\begin{definition}[Stacked Dispersed Graphs \cite{Alwen2017}] A DAG $G\in \mathbb{D}_{\varepsilon,g}^{\lambda,k}$ 
if there exist $\lambda\in\mathbb{N}^+$ disjoint subsets of nodes $\{L_i\subseteq V\}$, each of size $k$,
with the following two properties:
\begin{enumerate}
 \item For each $L_i$ there is a path running through all nodes of $L_i$.
 \item Fix any topological ordering of nodes of $G$. For each $i\in[\lambda]$ let $G_i$ be the sub-graph of $G$ 
 containing all nodes of $G$ up to the last node of $L_i$. Then $G_i$ is an $(\varepsilon, g, k)-$dispersed graph.
\end{enumerate}
\end{definition}
\noindent
For stacked dispersed graphs we have
\begin{theorem}[Thm. 4 in \cite{Alwen2017}]\label{thm:thm6}
Let $G$ be a DAG such that $ G \in \mathcal{D}_{\varepsilon, g}^{\lambda, k} $. 
Then we have 
$$   \picc{G} \geq  \varepsilon \lambda g (\frac{k}{2} - g).$$
\end{theorem}

\section{Riffle Scrambler}\label{sec:rifflePass}
The \rifflePass function uses the following parameters:
\begin{itemize}
 \item $s$ is a salt which is used to generate a graph $\GG$,
 \item $g$ - garlic, 
$\GG = \tuple{V, E}$, \ie $V = V_0 \cup \ldots \cup V_{2 \lambda g}$, $|V_i| = 
2^g$,
 \item $\lambda$ - is the number of layers of the graph $\GG$.
\end{itemize}

Let $HW(x)$ denote the Hamming weight of a binary string $x$ (\ie the number of ones in $x$), and $\bar x$ denotes a coordinate-wise complement of $x$ (thus  $HW(\bar{x})$
    denotes number of zeros in $x$).


\begin{definition}
 Let $B = (b_0 \ldots b_{n-1})\in\{0,1\}^n$ (a binary word of length $n$).
We define a rank $r_{B}(i)$ of a bit $i$ in $B$ as
\[ r_B(i) = |\{j < i: b_j = b_i\}|.\]
\end{definition}

\begin{definition}[Riffle-Permutation]
Let $B = (b_0 \ldots b_{n-1})$ be a binary word of length $n$.
A permutation $\pi$  induced by $B$ is defined as
\[
\pi_{B}(i) = \left\{\begin{array}{l l}
		     r_B(i) & \textrm{if } b_i = 0,\\
		     r_B(i) + HW(\bar{B}) & \textrm{if } b_i = 1\\
                    \end{array}\right.
\]
for all $0\leq 
i \leq n-1$.
\end{definition}

\begin{example}\label{ex:1}
Let $B = 11100100$, then $r_B(0) = 0$, $r_B(1) = 1$, $r_B(2) = 2$, $r_B(3) = 0$,
$r_B(4) = 1$, $r_B(5) = 3$, $r_B(6) = 2$, $r_B(7) = 3$.
 The Riffle-Permutation induced by $B$ is equal to 
$ \pi_B = \left(\begin{array}{c c c c c c c c}
         0 & 1 & 2 & 3 & 4 & 5 & 6 & 7\\
          4 & 5 & 6 & 0 & 1 & 7 & 2 & 3\\
         \end{array}\right)$. For graphical illustration of this example see Fig. \ref{fig:perm}.

\tikzstyle{block3} = [draw,fill=gray!20,minimum size=0.5em]
\begin{figure}
 \center
 \begin{tikzpicture}[->,>=stealth',shorten >=1pt,node distance=2.5cm,auto,main node/.style={rectangle,rounded corners,draw,align=center}]

\def \n {8}

\foreach \y in {2,3}
{
  \foreach \x in {1,...,\n}
  {
    \pgfmathtruncatemacro{\lx}{\x - 1}
    \pgfmathtruncatemacro{\ly}{\y - 2}
    \pgfmathtruncatemacro{\label}{\x\y}
    \node[block3,draw, circle] (\label) at ({1.3*\x}, {-1.7*\y}) 
{$v_{\lx}^{\ly}$};
  }
}

\draw[node,->] 		(12)   to (53);

\draw[node,->] 		(22)   to (63);

\draw[node,->] 		(32)   to (73);

\draw[node,->] 		(42)   to (13);

\draw[node,->] 		(52)   to (23);

\draw[node,->] 		(62)   to (83);

\draw[node,->] 		(72)   to (33);

\draw[node,->] 		(82)   to (43);

\end{tikzpicture}
\caption{A graph of Riffle-Permutation induced by  $B = 11100100$.}\label{fig:perm}
\end{figure}
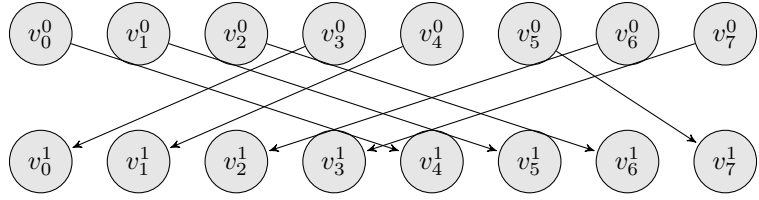

\end{example}

\begin{definition}[$N$-Single-Layer-Riffle-Graph]
 Let $\mathcal{V} = \mathcal{V}^0 \cup \mathcal{V}^1$ where  
$\mathcal{V}^i = \{v_{0}^i, \ldots, v_{N-1}^i\}$ and let $B$ be
an $N$-bit word. Let $\pi_B$ be a Riffle-Permutation induced by $B$.
We define $N$-Single-Layer-Riffle-Graph (for even $N$) via the set of edges 
$E$:
\begin{itemize}
 \item $1$ edge: $v_{N-1}^0 \rightarrow v_{0}^1$,
 \item $N$ edges: $v_i^0 \rightarrow v_{\pi_B(i)}^1$ for $i = 0, \ldots, N-1$,
 \item $N$ edges: $v_i^0 \rightarrow v_{\pi_{\bar{B}}(i) }^1$ for 
$i 
= 0, \ldots, N-1$.
\end{itemize}
\end{definition}

\begin{example}
 Continuing with $B$ and $\pi_B$ from Example \ref{ex:1}, the $8$-Single-Layer-Riffle-Graph is 
 presented in Fig. \ref{fig:layer}.
\tikzstyle{block3} = [draw,fill=gray!20,minimum size=0.5em]
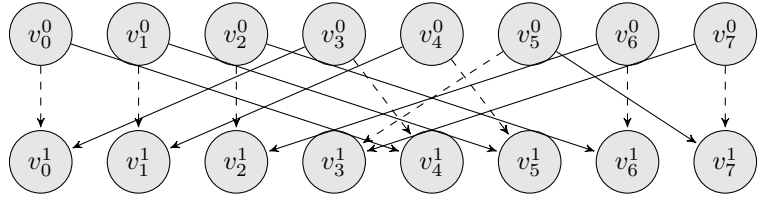
\begin{figure}
 \center
 \begin{tikzpicture}[->,>=stealth',shorten >=1pt,node distance=2.5cm,auto,main node/.style={rectangle,rounded corners,draw,align=center}]

\def \n {8}

\foreach \y in {2,3}
{
  \foreach \x in {1,...,\n}
  {
    \pgfmathtruncatemacro{\lx}{\x - 1}
    \pgfmathtruncatemacro{\ly}{\y - 2}
    \pgfmathtruncatemacro{\label}{\x\y}
    \node[block3,draw, circle] (\label) at ({1.3*\x}, {-1.7*\y}) 
{$v_{\lx}^{\ly}$};
  }
}


\draw[node,->] 		(12)   to (53);
\draw[node,dashed,->] 	(12)   to (13);

\draw[node,->] 		(22)   to (63);
\draw[node,dashed,->] 	(22)   to (23);

\draw[node,->] 		(32)   to (73);
\draw[node,dashed,->] 	(32)   to (33);

\draw[node,->] 		(42)   to (13);
\draw[node,dashed,->] 	(42)   to (53);

\draw[node,->] 		(52)   to (23);
\draw[node,dashed,->] 	(52)   to (63);

\draw[node,->] 		(62)   to (83);
\draw[node,dashed,->] 	(62)   to (43);

\draw[node,->] 		(72)   to (33);
\draw[node,dashed,->] 	(72)   to (73);

\draw[node,->] 		(82)   to (43);
\draw[node,dashed,->] 	(82)   to (83);

\end{tikzpicture}
\caption{An 8-Single-Layer-Riffle-Graph (with horizontal edges and edge $(v^0_7,v^1_0)$ skipped) for $B = 11100100$.
Permutation $\pi_B$ is depicted with solid lines, whereas $\pi_{\bar{B}}$ is presented with dashed lines.
%
}\label{fig:layer}
\end{figure}

\end{example}

%

%
%

Algorithm~\ref{alg:genConcentrator} is responsible for generating an
$N$-Double-Riffle-Graph  which is defined in the following way.
%
%
From now on, we assume that $N=2^g$.
\begin{definition}[$N$-Double-Riffle-Graph]
Let $V$ denote the set of vertices and $E$ be the set 
of edges of $\GG = (V, E)$. 
Let $B_0, \ldots, B_{g-1}$ be $g$ binary words of the length $2^g$ each.
Then $N$-Double-Riffle-Graph is obtained by stacking $2g$ 
Single-Layer-Riffle-Graphs resulting in a graph consisting of $(2g+1)2^g$ vertices
%
\begin{itemize}
 \item 
$ \{v_0^0, \ldots, v_{2^g-1}^0\} \cup \ldots \cup \{v_0^{2g}, \ldots, 
v_{2^g-1}^{2g}\},
$
\end{itemize}
and edges:
\begin{itemize}
 \item $(2g+1) 2^g$ edges: $v_{i-1}^j \rightarrow v_i^j$ for $i \in \{1, 
\ldots, 
2^g-1\}$ and $j \in \{0, 1, \ldots, 2^g\}$,
 \item $2g$ edges: $v_{2^g-1}^j \rightarrow v_{0}^{j+1}$ for $j \in \{0, 
\ldots, 2g-1\}$,
 \item $g 2^g$ edges: $v_i^{j-1} \rightarrow v_{\pi_{B_j}(i)}^j$ for $i \in \{ 0, 
\ldots, 2^g-1\}$, $j \in \{ 1, \ldots, g\}$,
 \item $g 2^g$ edges: $v_i^{j-1} \rightarrow v_{ \pi_{\bar{B_j}}(i) }^j $ for $i 
\in \{ 0, 
\ldots, 2^g-1\}$,  $j \in\{ 1, \ldots, g\}$,

\end{itemize}
and the following edges for the lower $g$ layers -- which are symmetric with 
respect to the level $g$ (and where we use inverse of permutations induced by $B_j ,j\in\{0,\ldots,g-1\}$):
\begin{itemize}
 \item $g 2^g$ edges: $v_{\pi^{-1}_{B_j}(i)}^{2g-j} \rightarrow v_i^{2g-j+1} 
\rightarrow $ for $i \in\{ 0, 
\ldots, 2^g-1\}$,  $j \in \{ 1, \ldots, g$\},
 \item $g 2^g$ edges: 
$v_{ i }^{2g-j} \rightarrow v_{\pi^{-1}_{B_j}(i) }^{2g-j+1} $
for $i \in\{ 0, 
\ldots, 2^g-1\}$,  $j \in \{ 1, \ldots, g\}$.

\end{itemize}

\end{definition}

\begin{definition}[$(N, \lambda)$-Double-Riffle-Graph]
Let $\GG_i, i=0,\ldots, \lambda-1$ be $N$-Double-Riffle-Graphs. 
The $(N, \lambda)$-Double-Riffle-Graph is a graph obtained by stacking $\lambda$ 
times $N$-Double-Riffle-Graphs together by joining outputs of $\GG_i$ to the 
corresponding inputs of $\GG_{i+1}, i=0,\ldots,\lambda-2$.
\end{definition}


One of the ingredients of our main procedure is a construction of $(N,\lambda)$-Double-Riffle-Graph 
using specific binary words $B_0,\ldots, B_{g-1}$.
To continue, we have to 
introduce ''trajectory tracing``. For given  $\sigma$ -- a permutation of $\{0,\ldots,2^g-1\}$ -- let 
$\mathbf{B}$ be a binary matrix of size $2^g\times g$, whose $j$-th row is a binary representation of
$\sigma(j), j=0,\ldots,2^g-1$.
Denote by $B_i$ the $i$-th column of $\mathbf{B}$, \ie $\mathbf{B}=(B_0,\ldots, B_{2^g-1})$.
We call this matrix a \textsl{binary representation of $\sigma$}.
Then we create a binary matrix $\mathfrak{B}=(\mathfrak{B}_0,\ldots,\mathfrak{B}_{2^g-1})$ (also of 
size $2^g\times g$) in the following way. We set $\mathfrak{B}_0=B_0$. For $i=1,\ldots, 2^g-1$ we 
set $\mathfrak{B}_i=\pi_{\mathfrak{B}_{i-1}^T}(B_i)$. The procedure 
\traceTraj is given in Algorithm \ref{alg:TraceTrajectories}.

  Roughly speaking, the procedure  $\rifflePass(x, s, g, \lambda)$    works in the following way. 
\begin{itemize}
 \item For given salt $s$ it  calculates a pseudorandom permutation $\sigma$ (using inverse Riffle Shuffle),
 let $\mathbf{B}$ be its binary representation.
 \item It calculates $\mathfrak{B}=\traceTraj(\mathbf{B})$.
  \item It creates an instance of $N$-Double-Riffle-Graph ($2g+1$ rows of $2^g$ vertices) using $\mathfrak{B}^T_0,\ldots, 
  \mathfrak{B}^T_{g-1}$ as binary words.
  \item It \textsl{evaluates} $x$ on the graph, calculates values at nodes on last row, \ie 
  $v_0^{2g+1},\ldots, v^{2^g-1}_{2g+1}$.
  \item Last row is rewritten to first one, \ie $v^0_i=v^{2g+1}_i, i=0,\ldots 2^g-1$,
  and whole evaluation is repeated $\lambda$ times.
  \item Finally, the value at $v^{2g}_{2^g-1}$ is returned.
%
\end{itemize}
The main procedure $\rifflePass(x, s, g, \lambda)$  
for storing a password $x$ using salt $s$ and 
memory-hardness parameters $g,\lambda$ is given in Algorithm \ref{alg:init}.
%

\begin{example}\label{ex:full}
An example of $(8, 1)$-Double-Riffle-Graph which was obtained from a 
permutation
$ \sigma = \left(\begin{array}{c c c c c c c c}
         0 & 1 & 2 & 3 & 4 & 5 & 6 & 7\\
         5 & 4 & 6 & 3 & 2 & 7 & 0 & 1\\
\end{array}\right)$.
Its binary representation is the following:
$$\mathbf{B}=(B_0, B_1, B_2), \quad \mathbf{B}^T=
\left( 
\begin{array}{cccccccccccccccc}
 1& 1& 1& 0& 0& 1& 0& 0 \\
 0& 0& 1& 1& 1& 1& 0& 0 \\
 1& 0& 0& 1& 0& 1& 0& 1\\
\end{array}
\right).
$$ 

%
%
%
\noindent
 We obtain trajectories of the elements
 and we can derive words/permutations for each layer of the graph:
\begin{itemize}
 \item 
$\mathfrak{B}_0=B_0 = (11100100)^T$
(used in the previous examples) -- obtained by concatenating first 
digits of elements, we have 
$ \pi_{B_0} = \left(\begin{array}{c c c c c c c c}
         0 & 1 & 2 & 3 & 4 & 5 & 6 & 7\\
          4 & 5 & 6 & 0 & 1 & 7 & 2 & 3\\
         \end{array}\right)$.
\item $\mathfrak{B}_1=\pi_{\mathfrak{B}_0^T}(B_1)=\pi_{\mathfrak{B}_0^T}(11100100)= (11000011)^T$,
thus 
%
%
$ \pi_{\mathfrak{B}_1} = \left(\begin{array}{c c c c c c c c}
         0 & 1 & 2 & 3 & 4 & 5 & 6 & 7\\
          4 & 5 & 0 & 1 & 2 & 3 & 6 & 7\\
         \end{array}\right)$.
\item $\mathfrak{B}_2=\pi_{\mathfrak{B}_1^T}(B_2)=\pi_{\mathfrak{B}_1^T}(10010101)=(01011001)^T$,
thus  $\pi_{B_2} = \left(\begin{array}{c c c c c c c c}
         0 & 1 & 2 & 3 & 4 & 5 & 6 & 7\\
          0 & 4 & 5 & 1 & 6 & 2 & 3 & 7\\
         \end{array}\right)$.
\item 
Finally, 
$$\mathfrak{B}=(\mathfrak{B}_0, \mathfrak{B}_1, \mathfrak{B}_2), \quad \mathfrak{B}^T=
\left( 
\begin{array}{cccccccccccccccc}
 1& 1& 1& 0& 0& 1& 0& 0 \\
 1& 1& 0& 0& 0& 0& 1& 1 \\
 0& 1& 0& 1& 1& 0& 0& 1\\
\end{array}
\right).
$$

\end{itemize}
The resulting graph is given in Fig. \ref{fig:main}.
\par 
\begin{figure}
\center
  \begin{subfigure}{4.0cm}
    \tikzstyle{reversed} = []
\begin{tikzpicture}[->,>=stealth',shorten >=1pt,node distance=2.5cm,auto,main node/.style={rectangle,rounded corners,draw,align=center},scale=0.37]

\def \n {8}

\foreach \y in {1,...,7}
{
  \foreach \x in {1,...,\n}
  {
    \pgfmathtruncatemacro{\label}{\x\y}
      \pgfmathtruncatemacro{\rx}{\x + 1}
    \node[scale=1.0] (\label) at ({1.3*\x}, {-2.7*\y}) {\textbullet};

  }
}

\foreach \y in {1,...,7}
{
  \foreach \x in {1,...,7}
  {
    \pgfmathtruncatemacro{\label}{\x\y}
      \pgfmathtruncatemacro{\rx}{\x + 1}
 \draw[node,shorten <=-5,shorten >=-3,densely dotted,->] ({\x\y})   to ( {\rx\y});  
  }
}

\foreach \y in {1,...,6}
{
  \pgfmathtruncatemacro{\ry}{\y + 1}
  \draw[node,shorten <=-5,shorten >=-3,densely dotted,->] ({8\y})   to ( {1\ry});
}
 
\end{tikzpicture}
    \caption{sequential + connecting layers (constant)}
  \end{subfigure}$\qquad\qquad$
  \begin{subfigure}{5.0cm}
    \tikzstyle{reversed} = []
\begin{tikzpicture}[->,>=stealth',shorten >=1pt,node distance=2.5cm,auto,main node/.style={rectangle,rounded corners,draw,align=center},scale=0.37]

\def \n {8}

\foreach \y in {1,...,7}
{
  \foreach \x in {1,...,\n}
  {
    \pgfmathtruncatemacro{\label}{\x\y}
      \pgfmathtruncatemacro{\rx}{\x + 1}
    \node[scale=1.0] (\label) at ({1.3*\x}, {-2.7*\y}) {\textbullet};

  }
}

%
%
%
%

%
%

\draw[node,shorten <=-5,shorten >=-3,->] 		(11)   to (12);
\draw[node,shorten <=-5,shorten >=-3,reversed,->] 		(11)   to (52);

\draw[node,shorten <=-5,shorten >=-3,->] 		(21)   to (62);
\draw[node,shorten <=-5,shorten >=-3,reversed,->] 		(21)   to (22);

\draw[node,shorten <=-5,shorten >=-3,->] 		(31)   to (72);
\draw[node,shorten <=-5,shorten >=-3,reversed,->] 		(31)   to (32);

\draw[node,shorten <=-5,shorten >=-3,->] 		(41)   to (12);
\draw[node,shorten <=-5,shorten >=-3,reversed,->] 		(41)   to (52);

\draw[node,shorten <=-5,shorten >=-3,->] 		(51)   to (22);
\draw[node,shorten <=-5,shorten >=-3,reversed,->] 		(51)   to (62);

\draw[node,shorten <=-5,shorten >=-3,->] 		(61)   to (82);
\draw[node,shorten <=-5,shorten >=-3,reversed,->] 		(61)   to (42);

\draw[node,shorten <=-5,shorten >=-3,->] 		(71)   to (32);
\draw[node,shorten <=-5,shorten >=-3,reversed,->] 		(71)   to (72);

\draw[node,shorten <=-5,shorten >=-3,->] 		(81)   to (42);
\draw[node,shorten <=-5,shorten >=-3,reversed,->] 		(81)   to (82);

\draw[node,shorten <=-5,shorten >=-3,->] 		(12)   to (53);
\draw[node,shorten <=-5,shorten >=-3,reversed,->] 		(12)   to (13);

\draw[node,shorten <=-5,shorten >=-3,->] 		(22)   to (63);
\draw[node,shorten <=-5,shorten >=-3,reversed,->] 		(22)   to (23);

\draw[node,shorten <=-5,shorten >=-3,->] 		(32)   to (13);
\draw[node,shorten <=-5,shorten >=-3,reversed,->] 		(32)   to (53);

\draw[node,shorten <=-5,shorten >=-3,->] 		(42)   to (23);
\draw[node,shorten <=-5,shorten >=-3,reversed,->] 		(42)   to (63);

\draw[node,shorten <=-5,shorten >=-3,->] 		(52)   to (33);
\draw[node,shorten <=-5,shorten >=-3,reversed,->] 		(52)   to (73);

\draw[node,shorten <=-5,shorten >=-3,->] 		(62)   to (43);
\draw[node,shorten <=-5,shorten >=-3,reversed,->] 		(62)   to (83);

\draw[node,shorten <=-5,shorten >=-3,->] 		(72)   to (73);
\draw[node,shorten <=-5,shorten >=-3,reversed,->] 		(72)   to (33);

\draw[node,shorten <=-5,shorten >=-3,->] 		(82)   to (83);
\draw[node,shorten <=-5,shorten >=-3,reversed,->] 		(82)   to (43);

\draw[node,shorten <=-5,shorten >=-3,->] 		(13)   to (14);
\draw[node,shorten <=-5,shorten >=-3,reversed,->] 		(13)   to (54);

\draw[node,shorten <=-5,shorten >=-3,->] 		(23)   to (54);
\draw[node,shorten <=-5,shorten >=-3,reversed,->] 		(23)   to (14);

\draw[node,shorten <=-5,shorten >=-3,->] 		(33)   to (64);
\draw[node,shorten <=-5,shorten >=-3,reversed,->] 		(33)   to (24);

\draw[node,shorten <=-5,shorten >=-3,->] 		(43)   to (24);
\draw[node,shorten <=-5,shorten >=-3,reversed,->] 		(43)   to (64);

\draw[node,shorten <=-5,shorten >=-3,->] 		(53)   to (74);
\draw[node,shorten <=-5,shorten >=-3,reversed,->] 		(53)   to (34);

\draw[node,shorten <=-5,shorten >=-3,->] 		(63)   to (34);
\draw[node,shorten <=-5,shorten >=-3,reversed,->] 		(63)   to (74);

\draw[node,shorten <=-5,shorten >=-3,->] 		(73)   to (44);
\draw[node,shorten <=-5,shorten >=-3,reversed,->] 		(73)   to (84);

\draw[node,shorten <=-5,shorten >=-3,->] 		(83)   to (84);
\draw[node,shorten <=-5,shorten >=-3,reversed,->] 		(83)   to (44);

\draw[node,shorten <=-5,shorten >=-3,->] 			(14) to (15)   ;
\draw[node,shorten <=-5,shorten >=-3,reversed,->] 		(54) to (15)   ;

\draw[node,shorten <=-5,shorten >=-3,->] 			(54) to (25)   ;
\draw[node,shorten <=-5,shorten >=-3,reversed,->] 		(14) to (25)   ;

\draw[node,shorten <=-5,shorten >=-3,->] 			(64) to (35)   ;
\draw[node,shorten <=-5,shorten >=-3,reversed,->] 		(24) to (35)   ;

\draw[node,shorten <=-5,shorten >=-3,->] 			(24) to (45)   ;
\draw[node,shorten <=-5,shorten >=-3,reversed,->] 		(64) to (45)   ;

\draw[node,shorten <=-5,shorten >=-3,->] 			(74) to (55)   ;
\draw[node,shorten <=-5,shorten >=-3,reversed,->] 		(34) to (55)   ;

\draw[node,shorten <=-5,shorten >=-3,->] 			 (34) to (65)   ;
\draw[node,shorten <=-5,shorten >=-3,reversed,->] 		(74) to (65)   ;

\draw[node,shorten <=-5,shorten >=-3,->] 			(44) to (75)   ;
\draw[node,shorten <=-5,shorten >=-3,reversed,->] 		(84) to (75)   ;

\draw[node,shorten <=-5,shorten >=-3,->] 			(84) to (85)   ;
\draw[node,shorten <=-5,shorten >=-3,reversed,->] 		(44) to (85)  ;

\draw[node,shorten <=-5,shorten >=-3,->] 			(55) to (16)   ;
\draw[node,shorten <=-5,shorten >=-3,reversed,->] 		(15) to (16)   ;

\draw[node,shorten <=-5,shorten >=-3,->] 			(65) to (26)   ;
\draw[node,shorten <=-5,shorten >=-3,reversed,->] 		(25) to (26)   ;

\draw[node,shorten <=-5,shorten >=-3,->] 			(15) to (36)   ;
\draw[node,shorten <=-5,shorten >=-3,reversed,->] 		(55) to (36)   ;

\draw[node,shorten <=-5,shorten >=-3,->] 			(25) to (46)   ;
\draw[node,shorten <=-5,shorten >=-3,reversed,->] 		(65) to (46)   ;

\draw[node,shorten <=-5,shorten >=-3,->] 			(35) to (56)   ;
\draw[node,shorten <=-5,shorten >=-3,reversed,->] 		(75) to (56)   ;

\draw[node,shorten <=-5,shorten >=-3,->] 			(45) to (66)   ;
\draw[node,shorten <=-5,shorten >=-3,reversed,->] 		(85) to (66)   ;

\draw[node,shorten <=-5,shorten >=-3,->] 			(75) to (76)   ;
\draw[node,shorten <=-5,shorten >=-3,reversed,->] 		(35) to (76)   ;

\draw[node,shorten <=-5,shorten >=-3,->] 			(85) to (86)   ;
\draw[node,shorten <=-5,shorten >=-3,reversed,->] 		(45) to (86)   ;


\draw[node,shorten <=-5,shorten >=-3,->] 			(16) to (17)   ;
\draw[node,shorten <=-5,shorten >=-3,reversed,->] 		(56) to (17)   ;

\draw[node,shorten <=-5,shorten >=-3,->] 			(66) to (27)   ;
\draw[node,shorten <=-5,shorten >=-3,reversed,->] 		(26) to (27)   ;

\draw[node,shorten <=-5,shorten >=-3,->] 			(76) to (37)   ;
\draw[node,shorten <=-5,shorten >=-3,reversed,->] 		(36) to (37)   ;

\draw[node,shorten <=-5,shorten >=-3,->] 			(16) to (47)   ;
\draw[node,shorten <=-5,shorten >=-3,reversed,->] 		(56) to (47)   ;

\draw[node,shorten <=-5,shorten >=-3,->] 			(26) to (57)   ;
\draw[node,shorten <=-5,shorten >=-3,reversed,->] 		(66) to (57)   ;

\draw[node,shorten <=-5,shorten >=-3,->] 			(86) to (67)   ;
\draw[node,shorten <=-5,shorten >=-3,reversed,->] 		(46) to (67)   ;

\draw[node,shorten <=-5,shorten >=-3,->] 			(36) to (77)   ;
\draw[node,shorten <=-5,shorten >=-3,reversed,->] 		(76) to (77)   ;

\draw[node,shorten <=-5,shorten >=-3,->] 			(46) to (87)   ;
\draw[node,shorten <=-5,shorten >=-3,reversed,->] 		(86) to (87)   ;

\end{tikzpicture}
    \caption{Main part of $\textsf{GenGraph}$  $\quad\qquad\quad\qquad$ \ (salt dependent)}
    \hspace{2cm}
  \end{subfigure}
  \begin{subfigure}{14.0cm}
    \tikzstyle{reversed} = []
\begin{tikzpicture}[->,>=stealth',shorten >=1pt,node distance=2.5cm,auto,main node/.style={rectangle,rounded corners,draw,align=center},scale=0.7]

\def \n {8}

\foreach \y in {1,...,7}
{
  \foreach \x in {1,...,\n}
  {
    \pgfmathtruncatemacro{\label}{\x\y}
      \pgfmathtruncatemacro{\rx}{\x + 1}
    \node[scale=1.2] (\label) at ({1.55*\x}, {-2.4*\y}) {\textbullet};

  }
}

\foreach \y in {1,...,7}
{
  \foreach \x in {1,...,7}
  {
    \pgfmathtruncatemacro{\label}{\x\y}
      \pgfmathtruncatemacro{\rx}{\x + 1}
 \draw[node,shorten <=-5,shorten >=-3,densely dotted,->] ({\x\y})   to ( {\rx\y});  
 
  }
}

\foreach \y in {1,...,6}
{
  \foreach \x in {1,...,\n}
  {
    \pgfmathtruncatemacro{\label}{\x\y}
      \pgfmathtruncatemacro{\rx}{\x + 1}
      \pgfmathtruncatemacro{\dy}{\y + 1}
      

  }
}

\foreach \y in {1,...,5}
{
  \pgfmathtruncatemacro{\ry}{\y + 1}
  \draw[node,shorten <=-5,shorten >=-3,densely dotted,->] ({8\y})   to ( {1\ry});
}

\node[scale=1.2] (inp) at (1.5, -1.7) {$v_0^0$};
\node[scale=1.2] (inp) at (6.3, -1.7) {$v_i^0$};
\node[scale=1.2] (inp) at (8.0, -1.7) {$v_{i+1}^0$};
\node[scale=1.2] (inp) at (12.7, -1.7) {$v_{2^g-1}^0$};

\node[scale=1.2] (inp) at (0.8, -9.4) {$v_0^{j}$};
\node[scale=1.2] (inp) at (13.2, -9.4) {$v_{2^g-1}^{j}$};

\node[scale=1.2] (inp) at (0.8, -12.0) {$v_0^{j+1}$};
\node[scale=1.2] (inp) at (13.2, -12.0) {$v_{2^g-1}^{j+1}$};

\node[scale=1.2] (inp) at (1.5, -17.4) {$v_0^{2g}$};
\node[scale=1.2] (inp) at (12.7, -17.4) {$v_{2^g-1}^{2g}$};

\node[scale=1.2] (inp) at (-0.5, -2.4) {{\small input}};
\draw[node,shorten <=-2,shorten >=-2,->] (inp)   to (11);

\node[scale=1.2] (outp) at (14.5, -16.8) {{\small output}};
\draw[node,->] (87)   to (outp);

\draw[node,shorten <=-5,shorten >=-3,->] 		(11)   to (12);
\draw[node,shorten <=-5,shorten >=-3,reversed,->] 		(11)   to (52);

\draw[node,shorten <=-5,shorten >=-3,->] 		(21)   to (62);
\draw[node,shorten <=-5,shorten >=-3,reversed,->] 		(21)   to (22);

\draw[node,shorten <=-5,shorten >=-3,->] 		(31)   to (72);
\draw[node,shorten <=-5,shorten >=-3,reversed,->] 		(31)   to (32);

\draw[node,shorten <=-5,shorten >=-3,->] 		(41)   to (12);
\draw[node,shorten <=-5,shorten >=-3,reversed,->] 		(41)   to (52);

\draw[node,shorten <=-5,shorten >=-3,->] 		(51)   to (22);
\draw[node,shorten <=-5,shorten >=-3,reversed,->] 		(51)   to (62);

\draw[node,shorten <=-5,shorten >=-3,->] 		(61)   to (82);
\draw[node,shorten <=-5,shorten >=-3,reversed,->] 		(61)   to (42);

\draw[node,shorten <=-5,shorten >=-3,->] 		(71)   to (32);
\draw[node,shorten <=-5,shorten >=-3,reversed,->] 		(71)   to (72);

\draw[node,shorten <=-5,shorten >=-3,->] 		(81)   to (42);
\draw[node,shorten <=-5,shorten >=-3,reversed,->] 		(81)   to (82);

\draw[node,shorten <=-5,shorten >=-3,->] 		(12)   to (53);
\draw[node,shorten <=-5,shorten >=-3,reversed,->] 		(12)   to (13);

\draw[node,shorten <=-5,shorten >=-3,->] 		(22)   to (63);
\draw[node,shorten <=-5,shorten >=-3,reversed,->] 		(22)   to (23);

\draw[node,shorten <=-5,shorten >=-3,->] 		(32)   to (13);
\draw[node,shorten <=-5,shorten >=-3,reversed,->] 		(32)   to (53);

\draw[node,shorten <=-5,shorten >=-3,->] 		(42)   to (23);
\draw[node,shorten <=-5,shorten >=-3,reversed,->] 		(42)   to (63);

\draw[node,shorten <=-5,shorten >=-3,->] 		(52)   to (33);
\draw[node,shorten <=-5,shorten >=-3,reversed,->] 		(52)   to (73);

\draw[node,shorten <=-5,shorten >=-3,->] 		(62)   to (43);
\draw[node,shorten <=-5,shorten >=-3,reversed,->] 		(62)   to (83);

\draw[node,shorten <=-5,shorten >=-3,->] 		(72)   to (73);
\draw[node,shorten <=-5,shorten >=-3,reversed,->] 		(72)   to (33);

\draw[node,shorten <=-5,shorten >=-3,->] 		(82)   to (83);
\draw[node,shorten <=-5,shorten >=-3,reversed,->] 		(82)   to (43);

\draw[node,shorten <=-5,shorten >=-3,->] 		(13)   to (14);
\draw[node,shorten <=-5,shorten >=-3,reversed,->] 		(13)   to (54);

\draw[node,shorten <=-5,shorten >=-3,->] 		(23)   to (54);
\draw[node,shorten <=-5,shorten >=-3,reversed,->] 		(23)   to (14);

\draw[node,shorten <=-5,shorten >=-3,->] 		(33)   to (64);
\draw[node,shorten <=-5,shorten >=-3,reversed,->] 		(33)   to (24);

\draw[node,shorten <=-5,shorten >=-3,->] 		(43)   to (24);
\draw[node,shorten <=-5,shorten >=-3,reversed,->] 		(43)   to (64);

\draw[node,shorten <=-5,shorten >=-3,->] 		(53)   to (74);
\draw[node,shorten <=-5,shorten >=-3,reversed,->] 		(53)   to (34);

\draw[node,shorten <=-5,shorten >=-3,->] 		(63)   to (34);
\draw[node,shorten <=-5,shorten >=-3,reversed,->] 		(63)   to (74);

\draw[node,shorten <=-5,shorten >=-3,->] 		(73)   to (44);
\draw[node,shorten <=-5,shorten >=-3,reversed,->] 		(73)   to (84);

\draw[node,shorten <=-5,shorten >=-3,->] 		(83)   to (84);
\draw[node,shorten <=-5,shorten >=-3,reversed,->] 		(83)   to (44);

\draw[node,shorten <=-5,shorten >=-3,->] 			(14) to (15)   ;
\draw[node,shorten <=-5,shorten >=-3,reversed,->] 		(54) to (15)   ;

\draw[node,shorten <=-5,shorten >=-3,->] 			(54) to (25)   ;
\draw[node,shorten <=-5,shorten >=-3,reversed,->] 		(14) to (25)   ;

\draw[node,shorten <=-5,shorten >=-3,->] 			(64) to (35)   ;
\draw[node,shorten <=-5,shorten >=-3,reversed,->] 		(24) to (35)   ;

\draw[node,shorten <=-5,shorten >=-3,->] 			(24) to (45)   ;
\draw[node,shorten <=-5,shorten >=-3,reversed,->] 		(64) to (45)   ;

\draw[node,shorten <=-5,shorten >=-3,->] 			(74) to (55)   ;
\draw[node,shorten <=-5,shorten >=-3,reversed,->] 		(34) to (55)   ;

\draw[node,shorten <=-5,shorten >=-3,->] 			 (34) to (65)   ;
\draw[node,shorten <=-5,shorten >=-3,reversed,->] 		(74) to (65)   ;

\draw[node,shorten <=-5,shorten >=-3,->] 			(44) to (75)   ;
\draw[node,shorten <=-5,shorten >=-3,reversed,->] 		(84) to (75)   ;

\draw[node,shorten <=-5,shorten >=-3,->] 			(84) to (85)   ;
\draw[node,shorten <=-5,shorten >=-3,reversed,->] 		(44) to (85)  ;

\draw[node,shorten <=-5,shorten >=-3,->] 			(55) to (16)   ;
\draw[node,shorten <=-5,shorten >=-3,reversed,->] 		(15) to (16)   ;

\draw[node,shorten <=-5,shorten >=-3,->] 			(65) to (26)   ;
\draw[node,shorten <=-5,shorten >=-3,reversed,->] 		(25) to (26)   ;

\draw[node,shorten <=-5,shorten >=-3,->] 			(15) to (36)   ;
\draw[node,shorten <=-5,shorten >=-3,reversed,->] 		(55) to (36)   ;

\draw[node,shorten <=-5,shorten >=-3,->] 			(25) to (46)   ;
\draw[node,shorten <=-5,shorten >=-3,reversed,->] 		(65) to (46)   ;

\draw[node,shorten <=-5,shorten >=-3,->] 			(35) to (56)   ;
\draw[node,shorten <=-5,shorten >=-3,reversed,->] 		(75) to (56)   ;

\draw[node,shorten <=-5,shorten >=-3,->] 			(45) to (66)   ;
\draw[node,shorten <=-5,shorten >=-3,reversed,->] 		(85) to (66)   ;

\draw[node,shorten <=-5,shorten >=-3,->] 			(75) to (76)   ;
\draw[node,shorten <=-5,shorten >=-3,reversed,->] 		(35) to (76)   ;

\draw[node,shorten <=-5,shorten >=-3,->] 			(85) to (86)   ;
\draw[node,shorten <=-5,shorten >=-3,reversed,->] 		(45) to (86)   ;


\draw[node,shorten <=-5,shorten >=-3,->] 			(16) to (17)   ;
\draw[node,shorten <=-5,shorten >=-3,reversed,->] 		(56) to (17)   ;

\draw[node,shorten <=-5,shorten >=-3,->] 			(66) to (27)   ;
\draw[node,shorten <=-5,shorten >=-3,reversed,->] 		(26) to (27)   ;

\draw[node,shorten <=-5,shorten >=-3,->] 			(76) to (37)   ;
\draw[node,shorten <=-5,shorten >=-3,reversed,->] 		(36) to (37)   ;

\draw[node,shorten <=-5,shorten >=-3,->] 			(16) to (47)   ;
\draw[node,shorten <=-5,shorten >=-3,reversed,->] 		(56) to (47)   ;

\draw[node,shorten <=-5,shorten >=-3,->] 			(26) to (57)   ;
\draw[node,shorten <=-5,shorten >=-3,reversed,->] 		(66) to (57)   ;

\draw[node,shorten <=-5,shorten >=-3,->] 			(86) to (67)   ;
\draw[node,shorten <=-5,shorten >=-3,reversed,->] 		(46) to (67)   ;

\draw[node,shorten <=-5,shorten >=-3,->] 			(36) to (77)   ;
\draw[node,shorten <=-5,shorten >=-3,reversed,->] 		(76) to (77)   ;

\draw[node,shorten <=-5,shorten >=-3,->] 			(46) to (87)   ;
\draw[node,shorten <=-5,shorten >=-3,reversed,->] 		(86) to (87)   ;

\end{tikzpicture}
    \caption{Final graph generated by  $\textsf{GenGraph}$}
  \end{subfigure}
  \caption{Instance of (8,1)-Double-Riffle-Graph  the graph}\label{fig:main}
\end{figure}
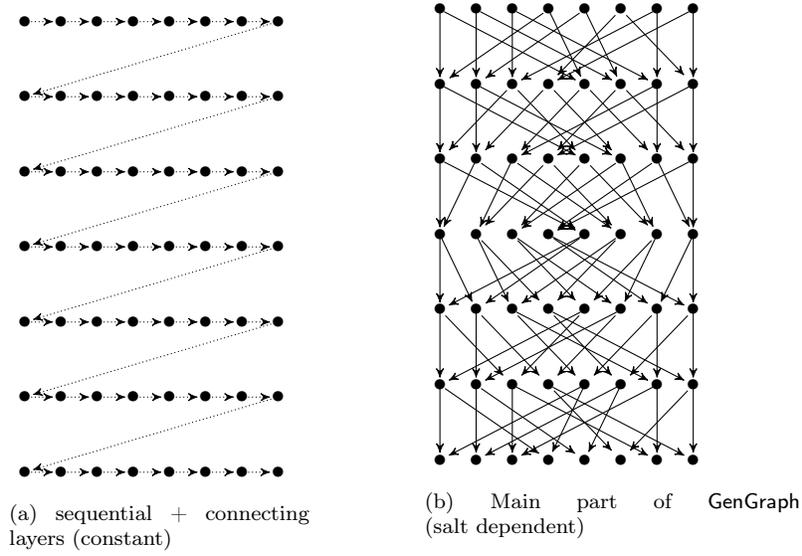
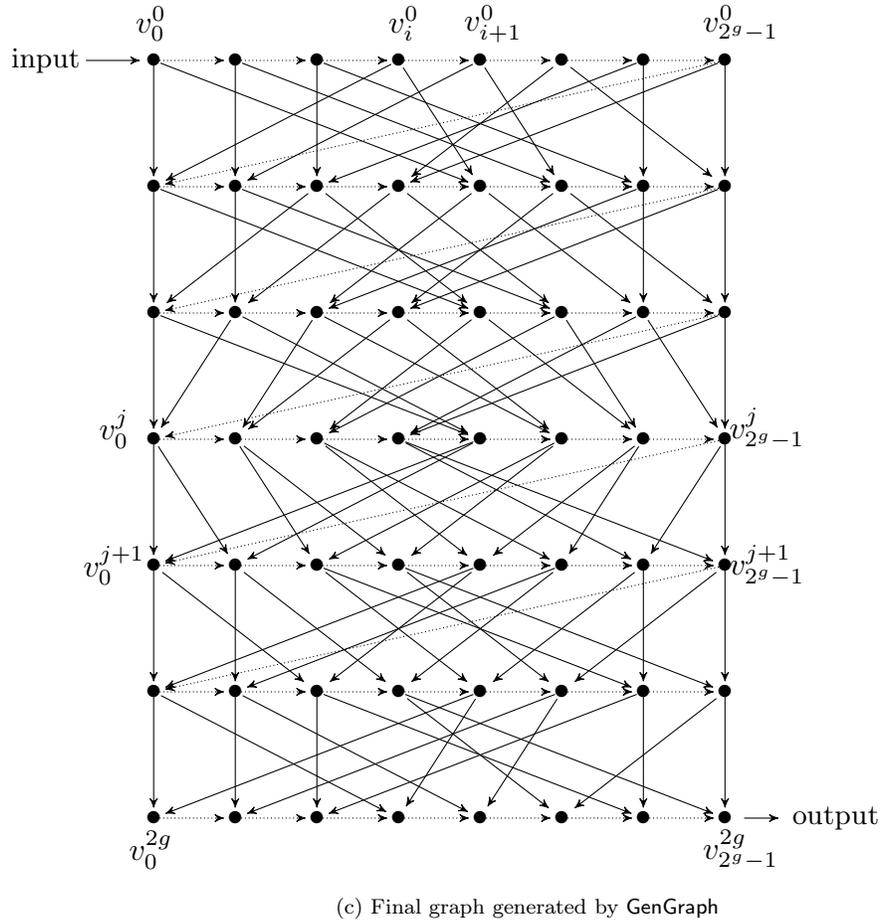


\end{example}

\subsection{Pseudocodes}
\begin{algorithm}[H]
\caption{$\textsf{RiffleShuffle}_H(n, s)$}\label{alg:riffle}
\begin{algorithmic}[1]
\State $\pi = \tuple{1, \ldots, n}$
\For{$i = 1$ to $n$}
  \For{$j = 1$ to $i-1$}
  \State $M[i, j] = 0$
  \EndFor
\EndFor

\While{$\exists_{1 \leq i \leq n} \exists_{1\leq j < i} M[i,j] = 0$}
  \State $\mathcal{S}_0, \mathcal{S}_1 = \emptyset$
  \For{$w:=1$ to $n$}
    \State $b = H(s, w, r)_1$
    \State $\mathcal{S}_{b} = \mathcal{S}_{b} \uplus \pi[w]$
  \EndFor
  \For{$i \in \mathcal{S}_0$}
    \For{$j \in \mathcal{S}_1$}
      \State $M[\max(i, j), \min(i,j)] = 1$
  \EndFor
  \EndFor
  \State $\pi = \mathcal{S}_0 \uplus \mathcal{S}_1$
\EndWhile \\
\Return $\pi$
\end{algorithmic}
\end{algorithm}
%

\begin{algorithm}[H]
\caption{$\traceTraj(\mathbf{B})$}\label{alg:TraceTrajectories}
\begin{algorithmic}[1]
\Require: $\mathbf{B}=(B_0,\ldots,B_{g-1})$ $\{\textsf{binary matrix of size}$ $2^g\times g \ \textsf{with columns } B_i, i=0,\ldots, g-1\}$
 
\Ensure: $\mathfrak{B}$ $\{\textsf{binary matrix of size}$ $2^g\times g \textsf{ with recalculated trajectories}\}$
\State $\mathfrak{B}_0=B_0$
\For{$i:=1$ to $g-1$}
\State $\mathfrak{B}_i=\pi_{\mathfrak{B}^T_{i-1}}(B_i^T)$ $\qquad // \textsl{Riffle-Permutation induced by } B_i^T$
\EndFor \\
\Return $\mathfrak{B}=(\mathfrak{B}_0,\ldots,\mathfrak{B}_{g-1})$
\end{algorithmic}
\end{algorithm}

\begin{algorithm}[H]
\caption{$\textsf{GenGraph}_H(g, \sigma)$}\label{alg:genConcentrator}
\begin{algorithmic}[1]
\State $N = 2^g$
\State $V = \{v_i^j: i = 0, \ldots, N-1; j = 0, \ldots, 2g \}$
\State $E = \{v_i^j \rightarrow v_{i+1}^j: i = 0, \ldots, N-2; j = 0, \ldots, 
2g\}$ 
\State $E = E \cup \{v_{n-1}^{j} \rightarrow v_{0}^{j+1}: j = 0, \ldots, 2g-1\}$
\State Let $\mathbf{B}$ be a binary representation of $\sigma$
\State Calculate $\mathfrak{B}=(\mathfrak{B}_0, \ldots, \mathfrak{B}_{g-1})=\traceTraj(\mathbf{B})$.
Let $\mathfrak{B}_{2g+1-m}=\mathfrak{B}_{m}, m=0,\ldots,g-2.$ Let $\mathfrak{B}_j=\mathfrak{B}_{j,0}\mathfrak{B}_{j,1}\ldots\mathfrak{B}_{j,2^g-1}$
\For{$i = 0$ to $2g$}
\EndFor
\For{$j = 0$ to $2g-1$}
  \For{$i = 0$ to $2^g-1$}
  \State $E = E \cup \{v_i^j \to v_{\pi_{\mathfrak{B}_{j,i}}}^{j+1}\} \cup \{v_i^j \to v_{\pi_{\bar{\mathfrak{B}}_{j,i}}}^{j+1}\}$ 
  \EndFor
\EndFor\\
\Return $\pi$
\end{algorithmic}
\end{algorithm}



\begin{algorithm}[H]
\caption{$\rifflePass(n, x, s, g, \lambda)$}\label{alg:init}
\begin{algorithmic}[1]
\Require  $s$ $\{\textsf{Salt}\}$, $g$ $\{\textsf{Garlic}\}$, $x$ 
$\{\textsf{Value to Hash}\}$, $\lambda$ 
$\{\textsf{Depth}\}$, $H$ $\{\textsf{Hash Function}\}$
\Ensure  $x$ $\{\textsf{Password Hash}\}$
\State $\sigma = \textsf{RiffleShuffle}_H(2^g, s)$
\State $G = (V, E) = \mathsf{GenGraph}(g, \sigma)$
\State $v_0^0 \leftarrow H(x)$
\For{$i:=1$ to $2^g-1$}
\State $v_i^0 = H(v_{i-1}^0)$ 
\EndFor
\For{$r:=1$ to $\lambda$}
  \For{$j:=0$ to $2g$}
    \For{$i=0$ to $2^g-1$}
    \State $v_i^{j+1} :=0$
      \ForAll{$v \rightarrow v_i^{j+1} \in E}$
	\State $v_i^{j+1} := H(v_i^{j+1}, v)$
      \EndFor
    \EndFor
  \EndFor     
      \For{$i=0$ to $2^g-1$}
	  \State $v_i^0=v_i^{2g+1}$
      \EndFor
\EndFor
\State $x := v_{n-1}^{2g}$\\
\Return $x$
\end{algorithmic}
\end{algorithm}

\section{Proofs}
\subsection{Proof of Lemma \ref{lem:result}}


In this Section  we present the security proof of our scheme in the sequential model of 
adversary.  Recall that $N=2^g$.
To prove Lemma \ref{lem:result} it is enough to prove the following theorem 
(since then 
the assertion follows from Theorem \ref{thm:superconcentrator}).

%

\begin{theorem}\label{thm:riffleConcentrator}
Let $\rho = (\rho_0, \ldots, \rho_{2^g-1}$) be a permutation of $N=2^g$ elements,
let $\mathbf{B}$ be its binary representation and let $\mathfrak{B}=
(\mathfrak{B}_0,\ldots,\mathfrak{B}_{g-1})=\traceTraj(\mathbf{B}
)$.
Let $\GG$ be an $N$-Double-Riffle-Graph using $\mathfrak{B}$. Then $\GG$ is an
$N$-Superconcentrator. 
\end{theorem}

Before we proceed to the main part of the proof, we shall introduce some auxiliary lemmas showing
some useful properties of the $N$-Double-Riffle-Graph.

%
\begin{lemma}
\label{lem:reversedPaths}
	Let $\rho$ be a permutation of $N = 2^g$ elements, $\rho = (\rho_0, \ldots, \rho_{2^g-1})$ and 
	let $\mathbf{B}$ be its binary representation. Let $\mathfrak{B}=\traceTraj(\mathbf{B})$.
	Let $\bar{\GG}$ be the subgraph
	of $N$-Double-Riffle-Graph $\GG$ constructed using $\mathfrak{B}$,
	consisting of $g+1$ layers and only of directed edges corresponding to the
	trajectories defined by $\rho$. Then the index of the endpoint of the directed path 
	from $j\mathrm{th}$
	input vertex $v_{j}^{0}$ of $\bar{\GG}$ is uniquely given by the reversal of the bit sequence
$\mathfrak{B}_j^T$. The input vertex $v_{j}^{0}$ corresponds to the output vertex $v_{k}^{g}$, where 
$k = \mathfrak{B}_j(2^g-1) \ldots \mathfrak{B}_j(1) \mathfrak{B}_j(0)$. 
\end{lemma}

The above lemma is very closely related to the interpretation of time-reversed Riffle Shuffle process
and it follows from a simple inductive argument. Let us focus on the last layer of $\bar{\GG}$, \ie the sets of
vertices $V_{g-1} = \{v_0^{g-1}, \ldots, v_{2^{g}-1}^{g-1}\}$, $V_{g} = \{v_0^{g}, \ldots, v_{2^{g}-1}^{g}\}$ and the
edges between them. Let us observe that from the construction of Riffle-Graph, the vertices from $V_{g}$ with indices whose 
binary representation starts with 1 are connected with the vertices from $V_{g-1}$ having the last bit of their trajectories
equal to 1 (the same applies to 0). Applying this kind of argument iteratively to the preceding layers will lead to the claim. 

Let us observe that choosing two different permutations $\rho_1$ and $\rho_2$ of $2^{g}$ elements will determine two different
sets of trajectories, each of them leading to different configurations of output vertices in $\bar{\GG}$ connected to
respective inputs from $V_{0} = \{v_0^{0}, \ldots, v_{2^{g}-1}^{0}\}$. Thus, as a simple consequence of
Lemma~\ref{lem:reversedPaths} we obtain the following corollary.
\begin{corollary}
	\label{cor:permVertices}
	There is one-to-one correspondence between the set $\mathbb{S}_2^{g}$ of permutations of $2^g$ elements determining the
	trajectories of input vertices in $N$-Double-Riffle-Graph and the positions of endpoint vertices from $V_{g}$ 
	connected with consecutive input vertices from $V_{0}$ in the graph $\bar{\GG}$ defined in Lemma \ref{lem:reversedPaths}.
\end{corollary}  

In order to provide a pebbling-based argument on memory hardness of the \rifflePass, we need the result on the structure
of $N$-Double-Riffle-Graph given by the following Lemma~\ref{lem:pairingButterfly}, which is somewhat similar to that of
Lemma~4 from \cite{Bradley2017}.
\begin{lemma}
\label{lem:pairingButterfly}
	Let $\GG$ be an $N$-Double-Riffle-Graph with $2^{g}$ input vertices. Then there is no such pairs of trajectories
	that if in some layer	$1 \leq i \leq g$ the vertices $v_{k}^{i-1}$ and $v_{l}^{i-1}$ on that trajectories are connected
	with some $v_{k'}^{i}$ and $v_{l'}^{i}$ (\ie they form a size-1 switch) on the corresponding trajectories, there is no
	$0 < j < i$ such that the predecessors (on a given trajectory) of $v_{k}^{i-1}$ and $v_{l}^{i-1}$ are connected with
	predecessors of $v_{k'}^{i}$ and $v_{l'}^{i}$.
\end{lemma}
\begin{proof}
	Assume towards contradiction that such two trajectories exist and denote by $v$ and $w$ the input vertices on that trajectories.
	For the sake of simplicity we will denote all the vertices on that trajectories by $v^{i}$ and $w^{i}$, respectively, where
	$i$ denotes the row of $\mathcal{G}$. Let $j$ and $k$ be the rows such that $v^j, v^{j+1}, w^{j}, w^{j+1}$ and
	$v^k, v^{k+1}, w^{k}, w^{k+1}$ are connected forming two size-1 switches. Without loss of generality, let $j < k$.
	Define the trajectories of $v$ and $w$ as
	\begin{itemize}
		\item $P_v = b_0^v b_1^v \ldots b_{j-1}^{v} b b_{j+1}^{v} \ldots b_{k-1}^{v} b' b_{k+1}^{v} \ldots b_{g}^{v}$ and
		\item $P_w = b_0^w b_1^w \ldots b_{j-1}^{w} \bar{b} b_{j+1}^{w} \ldots b_{k-1}^{w} \bar{b'} b_{k+1}^{w} \ldots b_{g}^{w}$.
	\end{itemize}
	Consider another two trajectories, namely
	\begin{itemize}
		\item $P_v^{*} = b_0^v b_1^v \ldots b_{j-1}^{v} \bar{b} b_{j+1}^{w} \ldots b_{k-1}^{w} \bar{b'} b_{k+1}^{v} \ldots b_{g}^{v}$
		      and
		\item $P_w^{*} = b_0^w b_1^w \ldots b_{j-1}^{w} b b_{j+1}^{v} \ldots b_{k-1}^{v} b' b_{k+1}^{w} \ldots b_{g}^{w}$.
	\end{itemize}
	Clearly, $P_v \neq P_v^{*}$ and $P_w \neq P_w^{*}$. Moreover, one can easily see that $P_v$ and $P_v^{*}$ move $v^0$ to $v^g$,
	$P_w$ and $P_w^{*}$ move $w^0$ to $w^g$. Since other trajectories are not affected when replacing $P_u$ with $P_u^{*}$
	for $u \in \{v,w\}$, such replacement does not change the order of vertices in $g^\mathrm{th}$ row.
	If $P_v$ and $P_w$ differ only on positions $j$ and $k$, then $P_v = P_w^{*}$ and $P_w = P_v^{*}$, what leads to two
	different sets of trajectories resulting in the same correspondence between vertices from
	rows 0 and $g$. This contradicts Corollary \ref{cor:permVertices}. Otherwise, there exist another two trajectories
	$P_x = P_v^{*}$ and $P_y = P_w^{*}$. However, in such situation, from Lemma~\ref{lem:reversedPaths} it follows that 
	$x^{0}$ and $v^{0}$ should be moved to the same vertex $x^{g} = v^{g}$, but this is not the case
	(similar holds for $y$ and $w$). Thus, the resulting contradiction finishes the proof.
\end{proof}

Being equipped with the lemmas introduced above, we can now proceed with the proof of Theorem 
\ref{thm:riffleConcentrator}.


\begin{proof}[of Theorem \ref{thm:riffleConcentrator}]
We need to show that for every $k: 1 \leq k \leq G$ and for every pair of 
subsets $V_{in} \subseteq \{v_0^0, \ldots, v_{2^g-1}^0\}$ and $V_{out} \subseteq
\{v_0^{2g}, \ldots, 
v_{2^g-1}^{2g}\}$, where $|V_{in}| = |V_{out}| = k$ there are $k$ 
vertex-disjoint paths connecting vertices in $V_{in}$ to the vertices in  
$V_{out}$.

The reasoning proceeds in a similar vein to the proofs of Theorem~1 and Theorem~2 from~\cite{Bradley2017}.
Below we present an outline of the main idea of the proof.

Let us notice that it is enough to show that for any $V_{in}$ of size $k$ there exists $k$ 
vertex-disjoint paths that connect vertices of $V_{in}$ to the vertices of 
$V_{middle} \subseteq \{v_0^{g}, \ldots, 
v_{2^g-1}^{g}\}$ with vertices of $V_{middle}$ forming a line, \ie either: 
\begin{itemize}
 \item $V_{middle} = \{v_i^g, v_{i+1}^g, \ldots, 
v_{i+k-1}^g\}$ for some $i$,
 \item or $V_{middle} = \{v_0^g, \ldots, v_{i-1}^g\} \cup \{v_{2^g-k+i}^g, 
\ldots, v_{2^g-1}^g\}$.
\end{itemize}
If the above is shown, we obtain the claim by finding vertex-disjoint paths between $V_{in}$ and $V_{middle}$
and then from $V_{middle}$ to $V_{out}$ from the symmetry of $G$-Double-Riffle-Graph.

Fix $1 \leq k \leq G$ and let $V_{in}$ and $V_{out}$ be some given size-$k$ subsets of input and output vertices, 
respectively, as defined above.
From Lemma~\ref{lem:pairingButterfly} it follows that if some two vertices $x$ and $y$ are connected in $i^\mathrm{th}$
round forming a size-1 switch then they will never be connected again until round $g$. Thus, the paths from $x$ and $y$ can
move according to different bits in that round, hence being distinguished (they will share no common vertex). 
Having this in mind, we can construct the nodes-disjoint paths from $V_{in}$ to $V_{middle}$ in the following way.
Starting in vertices from $V_{in}$, the paths will move through the first $t = \left\lceil \log k \right\rceil$ rounds according to all
distinct $t$-bits trajectories. Then, after round $t$, we will choose some fixed $g-t$-sequence $\tau$ common for all $k$ paths
and let them follow according $\tau$. Lemma~\ref{lem:reversedPaths} implies that the $k$ paths resulting
from the construction described above after $g$ steps will eventually end up in some subset $V_{middle}$ of $k$ vertices forming a line,
whereas Lemma~\ref{lem:pairingButterfly} ensures that the paths are vertex-disjoint. 

\end{proof}
%

\section{Summary}
We presented a new memory hard function which can be used as a secure 
password-storing scheme. \rifflePass achieves better time-memory trade-offs 
than Argon2i and Balloon Hashing when the buffer size $n$ grows and the number 
of rounds $r$ is fixed (and the same as Catena-DBG).

\begin{figure}
\scriptsize
\begin{tabular}{l | r | r | r |  r | r}
 &$BHG_7$ & $BHG_3$ & Argon2i & Catena BFG & 
\rifflePass 
\\[2pt]
 \hline 
 
 Server - time T (for S = N) & $8 \lambda N$ & $4 \lambda N$ & $2 \lambda N$ & 
$4 
\lambda N$ & $3 \lambda N$\\[2pt]
 \hline
Attacker - time T ($S \leq \frac{N}{64}$) & 
$T \geq \frac{2^\lambda -1}{32 S}N^2$ & 
\multirow{2}{*}{$T \geq \frac{\lambda N^2}{32 S}$} & 
\multirow{2}{*}{$T \geq \frac{N^2}{1536S}$}& 
\multirow{2}{*}{$T \geq (\frac{\lambda N}{64 S})^\lambda N$} &
\multirow{2}{*}{$T \geq (\frac{\lambda N}{64 S})^\lambda N$} 
\\[2pt]
Attacker - time T ($\frac{N}{64} \leq S \leq \frac{N}{20}$) & 
 unknown & &
& 
&
\\[2pt]
\hline 
Salt-dependent graph  & yes & yes & yes & no & yes\\[2pt]
\end{tabular}
\caption{Comparison of the efficiency and security of BalloonHashing 
($BHG_3$ is BalloonHashing with $\delta = 3$ and $BHG_7$ is BHG graph for 
$\delta = 7$), Argon2i, Catena (with Butterfly graph) and 
\rifflePass (RSG).}
\end{figure}

On the other hand, in the case of massively-parallel adversaries the situation 
of Catena is much worse than \rifflePass. In Catena, only double-butterfly 
graph (Catena-DBG) offered a time-memory trade-off which depends on the number 
of layers of the graph. Unfortunately, Catena-DBG is just a single 
graph instance, so in theory an adversary may try to build a parallel hardware 
to achieve better trade-off. In the case of \rifflePass the time-memory 
trade-off is the same as for the Catena-DBG but for \rifflePass, a different 
salt corresponds (with high probability) to the computation on a different (one 
from $N!$) graph. So an 
adversary cannot build just a single device for cracking all passwords (as in 
the case of 
Catena).


\begin{thebibliography}{WCW{\etalchar{+}}17}

\bibitem[AB16]{Alwen2016}
Jo\"{e}l Alwen and Jeremiah Blocki.
\newblock {Efficiently Computing Data-Independent Memory-Hard Functions}.
\newblock pages 241--271. 2016.

\bibitem[AB17]{Alwen2017b}
Joel Alwen and Jeremiah Blocki.
\newblock {Towards Practical Attacks on Argon2i and Balloon Hashing}.
\newblock In {\em 2017 IEEE European Symposium on Security and Privacy
  (EuroS\&P)}, pages 142--157. IEEE, April 2017.

\bibitem[ABP17]{Alwen2017}
Jo\"{e}l Alwen, Jeremiah Blocki, and Krzysztof Pietrzak.
\newblock {Depth-Robust Graphs and Their Cumulative Memory Complexity}.
\newblock In {\em EUROCRYPT 2017}, pages 3--32. Springer, Cham, April 2017.

\bibitem[AD86]{Aldous1986}
David Aldous and Persi Diaconis.
\newblock {Shuffling cards and stopping times}.
\newblock {\em American Mathematical Monthly}, 93(5):333--348, 1986.

\bibitem[AD87]{Aldous1987}
David Aldous and Persi Diaconis.
\newblock {Strong Uniform Times and Finite Random Walks}.
\newblock {\em Advances in Applied Mathematics}, 97:69--97, 1987.

\bibitem[ant]{antminer}
{The Antminer S9 Bitcoin Miner - Bitmain}.

\bibitem[AS15]{Alwen2015}
Jo\"{e}l Alwen and Vladimir Serbinenko.
\newblock {High Parallel Complexity Graphs and Memory-Hard Functions}.
\newblock In {\em Proceedings of the Forty-Seventh Annual ACM on Symposium on
  Theory of Computing - STOC '15}, pages 595--603, New York, New York, USA,
  2015. ACM Press.

\bibitem[BCGS16]{Boneh2016}
Dan Boneh, Henry Corrigan-Gibbs, and Stuart Schechter.
\newblock {Balloon Hashing: A Memory-Hard Function Providing Provable
  Protection Against Sequential Attacks}.
\newblock pages 220--248. Springer, Berlin, Heidelberg, December 2016.

\bibitem[BDK16]{Biryukov2016}
Alex Biryukov, Daniel Dinu, and Dmitry Khovratovich.
\newblock {Argon2: New Generation of Memory-Hard Functions for Password Hashing
  and Other Applications}.
\newblock In {\em 2016 IEEE European Symposium on Security and Privacy
  (EuroS\&P)}, pages 292--302. IEEE, March 2016.

\bibitem[Bra17]{Bradley2017}
William~F. Bradley.
\newblock {Superconcentration on a Pair of Butterflies}.
\newblock January 2017.

\bibitem[Bur00]{pbkdf}
Kaliski Burt.
\newblock {PKCS \#5: Password-Based Cryptography Specification Version 2.0}.
\newblock Technical report, 2000.

\bibitem[FLW]{Forler}
Christian Forler, Stefan Lucks, and Jakob Wenzel.
\newblock {Catena: A Memory-Consuming Password-Scrambling Framework}.

\bibitem[FLW15]{Forler2015_Catena}
Christian Forler, Stefan Lucks, and Jakob Wenzel.
\newblock {The Catena Password-Scrambling Framework}.
\newblock 2015.

\bibitem[LT82]{Lengauer1982}
Thomas Lengauer and Robert~E. Tarjan.
\newblock {Asymptotically tight bounds on time-space trade-offs in a pebble
  game}.
\newblock {\em Journal of the ACM}, 29(4):1087--1130, October 1982.

\bibitem[Per]{Percival}
Colin Percival.
\newblock {STRONGER KEY DERIVATION VIA SEQUENTIAL MEMORY-HARD FUNCTIONS}.

\bibitem[phc]{phc}
{Password Hashing Competition}.

\bibitem[Por15]{Pornin2015}
Thomas Pornin.
\newblock {The MAKWA Password Hashing Function Specifications v1.1}.
\newblock 2015.

\bibitem[SAA{\etalchar{+}}]{Simplicio}
Marcos~A Simplicio, Leonardo~C Almeida, Ewerton~R Andrade, Paulo C~F {Dos
  Santos}, and Paulo S L~M Barreto.
\newblock {Lyra2: Password Hashing Scheme with improved security against
  time-memory trade-offs}.

\bibitem[WCW{\etalchar{+}}17]{Wang2017}
Ding Wang, Haibo Cheng, Ping Wang, Xinyi Huang, and Gaopeng Jian.
\newblock {Zipf’s Law in Passwords}.
\newblock {\em IEEE Transactions on Information Forensics and Security},
  12(11):2776--2791, November 2017.

\bibitem[WZW{\etalchar{+}}16]{Wang2016}
Ding Wang, Zijian Zhang, Ping Wang, Jeff Yan, and Xinyi Huang.
\newblock {Targeted Online Password Guessing}.
\newblock In {\em Proceedings of the 2016 ACM SIGSAC Conference on Computer and
  Communications Security - CCS'16}, pages 1242--1254, New York, New York, USA,
  2016. ACM Press.

\bibitem[yes14]{yescrypt}
{yescrypt -a Password Hashing Competition submission}.
\newblock Technical report, 2014.

\end{thebibliography}
\newcommand{\etalchar}[1]{$^{#1}$}

\appendix

\section{Markov chains, mixing times, strong stationary times}\label{sec:markov}
Consider an ergodic Markov chain $\X=\{X_k, k\geq 0\}$ on a finite state space 
$\E $ with stationary distribution $\pi$. Ergodicity implies that 
$d_{TV}(\mathcal{L}(X_k),\pi)\to 0$ (total variation distance) as $k\to\infty$, where $\mathcal{L}(X_k)$ is the distribution of the chain at step $k$. 
Define
$$\tau_{mix}(\varepsilon)=\inf\{k: d_{TV}(\mathcal{L}(X_k),\pi)\leq \varepsilon\},$$
often called  total variation \textsl{mixing time}. 
\smallskip\par 
\noindent\textbf{Perfect simulation and strong stationary times}.
\textsl{Perfect simulation} refers to the art of converting an algorithm
for running a Markov chain into an algorithm which returns an unbiased sample 
from its stationary distribution. \textsl{Coupling from the past} (CFTP)
is one of the most known ones. However, it is not applicable to our 
models. Another method is based on \textsl{strong stationary times} (SSTs), a 
method whose prime application is studying the rate of convergence 
(of a chain to its stationary distribution). 
\begin{definition}
A stopping time is a random variable $T\in\mathbb{N}$ such that the event $\{T=k\}$ depends only on $X_0,\ldots,X_k$. 
A stopping time $T$ is a \textbf{strong stationary time} (SST) if 
$$\forall(i\in\E) P(X_k = i | T = k ) = \pi(k).$$ 
\end{definition}
Due to \cite{Aldous1987},we have
\begin{equation}\label{eq:TV_sep_SST}
d_{TV}(\mathcal{L}(X_k),\pi)\leq Pr(T>k).
\end{equation}

\begin{remark}
\rm 
An SST $T$ can be \textsl{non-randomized}, i.e., an event $\{T=k\}$ depends 
only on on the path $X_0,\ldots, X_k$ and \textsl{does not} use any extra randomness
(which is the case for \textsl{randomized} SST). For non-randomized SST performing 
perfect simulation is relatively easy: simply run the chain $X_k$ until event $T$ occurs and stop.
Then $X_T$ has distribution $\pi$.
\end{remark}
\medskip\par 
\noindent\textbf{Obtaining a random permutation: SST for Riffle Shuffle.}
In our applications we will need a random permutation of $N$ elements. 
We can think of these elements as of cards, and of a chain on permutations as of card shuffling.
Consider the following shuffling scheme (one step of a corresponding Markov chain):
\begin{quote}
 \textsl{For each card in the deck, flip a coin and label the back of the card with 0 (if Heads occurred)
 or with 1 (if Tails occurred). Take all the cards labeled 1 out of the deck and put them 
 on the top {keeping} their relative ordering.}
\end{quote}
This is the \textsl{inverse Riffle Shuffle}. 
The following rule is an SST for this shuffling (due to Aldous and Diaconis \cite{Aldous1986}):
\begin{quote}
\textsl{Keep track of assigned bits (storing a sequence of assigned bits for each card). 
Stop when  all the sequences are different.}
\end{quote}
Once stopped, we have a  permutation obtained from \textsl{exactly} 
uniform distribution (which is the stationary distribution of this chain). 
On average, this takes $2 \log_2 N$ steps. What we considered was \textsl{an idealized model}, in a sense that if we use some shuffling scheme in cryptography,
we do not have an infinite sequence of random numbers. In reality we must obtain them in a deterministic way,
\eg from   a private key or a salt (depending on applications).
In our constructions, it depends on some salt $s$ and  some hash function $H$. 
The Algorithm \ref{alg:riffle} called $\textsf{RiffleShuffle}_H(n, s)$ takes as input a salt $s$, a hash function $H$ and the number of elements $n$.
It performs inverse Riffle Shuffle (with randomness obtained from salt $s$ and hash function $H$) until above SST event occurs.
Thus, its output is a random permutation of $\{1,\ldots,n\}$ (in the 
random-oracle model).


\section{Parallel attacks}
In this section we formulate lemmas which are used to prove Lemma 
\ref{lem:main_par}.
%

%
%

%
%

These lemmas follow  directly the proof of parallel security of Catena 
presented in~\cite{Alwen2017}. The proof technique is exactly the same as 
in~\cite{Alwen2017} since $\rsg$ is a Superconcentrator. The only difference 
comes from the fact that each stack of Catena (Butterfly) is built out of 
$2g-1$ layers while $\rsg$ is built from $2g$ layers.
First, let us collect some observations on $\rsg$.
\begin{lemma}\label{lem:par_rsg_uwagi}
Let $\lambda, n \in \mathbb{N}^{+}$ such that $n = n' (2 c \lambda +1)$ 
where $n' = 2^c$ for some $c \in \mathcal{N}^{+}$. Then the Riffle Shuffle 
graph $\rsg$ consists of a stack
$\lambda$ sub-graphs such that the following holds.
\begin{enumerate}
 \item The graph $\rsg$ has $n$ nodes in total.
 \item The graph $\rsg$ is built as a stack of $\lambda$ sub-graphs $\{G_i\}_{i 
\in [\lambda]}$ each of which is a Superconcentrator. In the unique topological 
ordering of $\rsg$ denote the first and final $n'$ nodes of each $G_i$ as 
$L_{i,0}$ and $L_{i,1}$ respectively. Then there is a path running through all 
nodes in each $L_{i,1}$.
\item Moreover, for any $i \in [\lambda]$ and subsets $S \subset L_{i,0}$ and 
$T \subset L_{i,1}$ with $|S| = |T| = h \leq n'$ there exist $h$ nodes-disjoint 
paths $p_1, \ldots, p_h$ of length $2c$ from $S$ to $T$.
\end{enumerate} 
\end{lemma}
Using Lemma \ref{lem:par_rsg_uwagi}  and  following the proof of Lemma 6 in \cite{Alwen2017} we  obtain the following lemma.
\begin{lemma}\label{lem:par_D}
Let $\lambda, n \in \mathbb{N}^{+}$ be such that $n = n' (2 \lambda c + 1)$ with 
$n' = 2^c$ for some $c \in \mathcal{N}^{+}$.
 Then it holds that $\rsg \in \mathcal{D}_{1, g}^{\lambda, n'}$ for $g = \lceil 
\sqrt{n'}\rceil$
\end{lemma}
\noindent
The above Lemma \ref{lem:par_rsg_uwagi} together with Theorem \ref{thm:thm6} imply  Lemma 
\ref{lem:main_par}.

%
%
%

%
%

\end{document}